\documentclass[journal]{IEEEtran}
\usepackage{cite,color}

\ifCLASSINFOpdf
\else
\fi

\usepackage{amsmath,amssymb,amsfonts}
\usepackage{algorithmic}
\usepackage{algorithm}
\usepackage{array}
\usepackage{url}

\usepackage{float}
\usepackage{subfigure,multirow,bm,stfloats}
\usepackage{textcomp}
\usepackage{amsthm}
\newtheorem{myDef}{Definition}

\newtheorem{theorem}{Theorem}

\usepackage{graphicx}

\begin{document}
%
\title{Hypergraph Spectral Analysis and Processing in 3D Point Cloud}
%
%
%
\author{Songyang~Zhang, Shuguang~Cui,~\IEEEmembership{Fellow,~IEEE},
        and~Zhi~Ding,~\IEEEmembership{Fellow,~IEEE}}


\maketitle

\begin{abstract}
Along with increasingly popular virtual reality applications, 
the three-dimensional (3D) point cloud 
has become a fundamental data structure to characterize 3D objects and surroundings. 
To process 3D point clouds efficiently, a suitable model for the underlying structure and outlier noises is always critical. In this work, we propose a hypergraph-based new 
point cloud model that is amenable to efficient analysis and processing. 
We introduce tensor-based methods to estimate hypergraph spectrum components 
and frequency coefficients of point clouds in both ideal and noisy settings. 
We establish an analytical connection between hypergraph frequencies and structural features.
We further evaluate the efficacy of hypergraph spectrum estimation in two common
point cloud applications of sampling and denoising for which also
we elaborate specific hypergraph filter design and spectral properties.
The empirical performance demonstrates the strength of hypergraph signal processing as a tool in 3D point clouds and the underlying properties.

\end{abstract}


\section{Introduction}\label{intro}
Recent developments in depth sensors and softwares make it easier to capture the features and create a three-dimensional (3D) model for an object and its surroundings\cite{c1}. In particular, with the low-cost scanners such as light detection and ranging (LIDAR) and Kinect, a new data structure known as the point cloud has achieved significant success in many areas, including virtual reality, geographic information system, reconstruction of art document and high-precision 3D maps for self-driving cars \cite{c2}. 
A point cloud consists of 3D coordinates with attributes such as color, 
temperature, texture, and depth \cite{c3}. Owing to the easy access to scanning
sensors and the huge need in describing the 3D features, the use of point clouds 
has attracted significant attentions in areas of computer vision, virtual reality, 
and medical science. How to process the point clouds efficiently 
becomes an important topic of research in many 3D imaging and vision systems.

To analyze the features of point cloud, the first step is to construct an analytical model 
to represent the 3D structures. The literature provides several different models. 
In \cite{c4}, the 3D space is partitioned into several boxes or voxels, and the point clouds are then discretized therein. One disadvantage of voxels is that a dense grid 
is required to achieve fine resolution, leading to spatial inefficiency~\cite{c3}. 
A spatially efficient approach \cite{c5,c6} is the octree representation of point clouds. 
An octree is a tree data structure in which each node has exactly eight children. 
It can partition a 3D space recursively, and represent the point clouds 
with partitioned boxes. Although efficient, octree suffers from 
discretization errors \cite{c3}. The bd-tree is another spatial 
decomposition technique and is robust in highly cluttered point 
cloud dataset. However, compared to octree structures, 
bd-trees are more difficult to update. 

Recently, graphs and graph signal processing (GSP) have
found applications in modeling point clouds. For example, the authors of
\cite{c3} construct a graph based on pairwise point distances. 
Some other works, such as \cite{c8}, construct graphs based on the $k$-nearest neighbors, 
where each vertex (point) has an edge connection to its $k$ nearest neighbors. 
There are several clear connections between graph features and point cloud characteristics. 
For example, the smoothness over a graph can describe the flatness of surfaces 
in point clouds.
GSP-based tools such as filters and graph learning methods can 
process the point clouds and have shown great success because of the graph model's 
ability to capture the underlying geometric structures. 
However, graph-based methods still face some challenges such as limited orders 
and measurement inefficiency. In a traditional graph, 
each edge can only connect two nodes, constraining graph-based models
to describe only pairwise relationships. However, a multilateral 
relationship among multiple nodes is far more informative as in a point cloud model. 
For example, the points (nodes) on the same surface of 
a point cloud exhibit a strong multilateral relationship, 
which cannot be  easily captured by an edge of a traditional graph. 
In fact, construction of an efficient graph for a given dataset is always an 
open question. Thus, studies on point clouds can benefit from more general
and efficient models. 

To develop an efficient model for point clouds, we explore a high-dimensional graph model, known as hypergraph \cite{c9}. Hypergraph can be a useful model in processing 3D point clouds.
A hypergraph $\mathcal{H}=\{\mathcal{V},\mathcal{E}\}$ 
consists of a set of nodes $\mathcal{V}=\{\mathbf{v}_1, \dots,\mathbf{v}_K\}$ and a set of 
hyperedges $\mathcal{E}=\{\mathbf{e}_1, \dots,\mathbf{e}_K\}$. Each hyperedge in a 
hypergraph can 
connect more than two nodes. For example, a 3D shape together with its hypergraph model 
are shown as Fig. \ref{hyper}. Obviously, a normal graph is a special case of hypergraph, 
where each hyperedge degrades to connect two nodes exactly. The hyperedge in a hypergraph can
characterize the multilateral relationship among several related nodes (e.g., on a surface),
thereby making hypergraph a natural and intuitive model for point clouds. 
Moreover, advances in hypergraph signal processing (HGSP) \cite{c9} are
providing more hypergraph tools, such as HGSP-based filters and spectrum analysis,  
for effective point cloud processing.  

\begin{figure}[htbp]
	\subfigure[A 3D Shape with Six Nodes.]{
	\label{hg1}		
	\centering
	\includegraphics[height=1.4in]{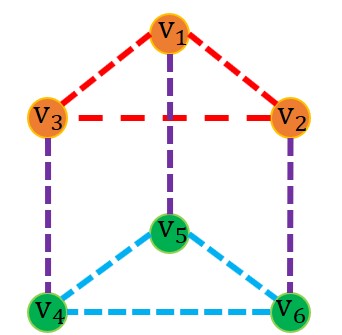}
}
	\subfigure[Hypergraph Model with 5 hyperedges and 6 Nodes:
	$\mathbf{e}_1=\{\mathbf{v}_2,\mathbf{v}_3,\mathbf{v}_4,\mathbf{v}_6\}$,
    $\mathbf{e}_2=\{\mathbf{v}_1,\mathbf{v}_3,\mathbf{v}_4,\mathbf{v}_5\}$,
    $\mathbf{e}_3=\{\mathbf{v}_1,\mathbf{v}_2,\mathbf{v}_3\}$,
    $\mathbf{e}_4=\{\mathbf{v}_4,\mathbf{v}_5,\mathbf{v}_6\}$,
    $\mathbf{e}_5=\{\mathbf{v}_1,\mathbf{v}_2,\mathbf{v}_5,,\mathbf{v}_6\}$.]{
	\label{hg2}		
	\centering
	\includegraphics[height=1.4in]{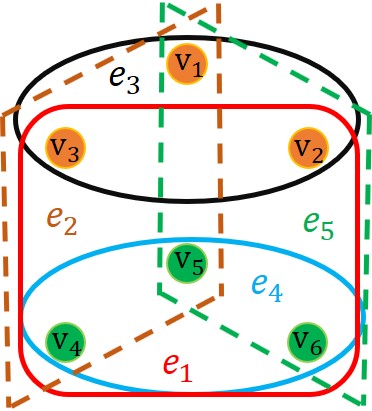}
}
\caption{Example of Hypergraphs.}
\label{hyper}
\end{figure}

However, processing the point clouds based on hypergraph still poses 
several challenges. Similar to GSP, the first problem lies in the 
construction of hypergraph for point clouds. 
The traditional hypergraph construction method for a general dataset relies on data structure. 
For example, in \cite{c11}, a hypergraph model is constructed according to 
the sentence structure in natural language processing. 
The $k$-nearest neighbor model is another method to construct the hypergraph. 
In \cite{c9}, a hypergraph can be formed from the feature distances for an 
animal dataset to achieve clustering. 
However, such distance-based or structure-based model may be rather lossy in  
information preservation. 
For example, the structure-based method may not preserve the correlation of 
some irregular structures, whereas the $k$-nearest neighbor method 
may narrowly emphasize the distance information. 
In addition to hypergraph construction, 
another issue in analyzing point cloud with hypergraph tools
is the computation complexity of the spectrum space. 
In the HGSP framework, spectrum-based analysis plays an important role but 
needs to compute the 
spectrum space. Usually, the computation of hypergraph spectrum 
is based on orthogonal-CP decomposition, which incurs high-complexity when
there are many nodes. Another challenge in point cloud processing
is the effect of noise and outliers. Since a hypergraph model is constructed 
from observed data, noise can distort the hypergraph and
degrade the performances of HGSP. Thus, mitigating
noise effect and robustly estimating the hypergraph model for point clouds
pose a significant challenge.

This work addresses the aforementioned problems. 
We propose novel spectrum-based hypergraph construction methods for 
both clean and noisy point clouds. For clean point clouds, we 
first estimate their spectrum components based on the hypergraph stationary process 
and optimally determine their frequency coefficients based on 
smoothness to recover the original hypergraph structure. 
For noisy point clouds, we introduce a method for joint hypergraph structure estimation and data denoising. 
We shall illustrate the effectiveness of the proposed
hypergraph construction and spectrum estimation in two point clould
applications: sampling and denoising. 
Our experimental results clearly establish 
a connection between hypergraph frequencies and point cloud features. 
The  performance improvement in both applications demonstrates
the strength and power of hypergraph in point cloud processing
and the practical value of our estimation methods.

We organize 
the rest of the paper as follows. In Section \ref{pre}, we lay the foundation
with respect to the 
preliminaries and notations of point clouds, tensor basics and hypergraph signal processing. Next, we propose means in estimating hypergraph spectrum 
for basic point clouds in Section \ref{h1} and further develop
means for hypergraph structure estimation of noisy point clouds in 
Section \ref{h2}. With the proposed estimation methods, we 
study two important application scenarios and establish the
effectiveness of hypergraph signal processing in Section \ref{appli}.
Finally, we present the conclusion and future directions in Section \ref{con}.

\section{Preliminaries and Notations}\label{pre}
In this section, we cover basic background with respect to 
point cloud, tensor basics and hypergraph signal processing. 

\subsection{Point Clouds}
A point cloud is a set of 3D points obtained from sensors, 
where each point is attributed with coordinates and other features, 
like colors \cite{c10}. Since the coordinates are  
basic features of a point cloud, in this work, we mainly focus on 
gray-scale point clouds, where each node is characterized by its coordinates. 
We consider a matrix representation of the gray-scale point clouds, 
where a point cloud with $N$ nodes is denoted by a location matrix
\begin{equation}
	\mathbf{s}=[\mathbf{X_1\quad X_2\quad X_3}]=
	\begin{bmatrix}
	\mathbf{s}_1^T\\
	\mathbf{s}_2^T\\
	\ddots\\
	\mathbf{s}_N^T
	\end{bmatrix}\in\mathbb{R}^{N\times 3},
\end{equation}
where $\mathbf{X}_i$ denotes a vector of the $i$th coordinates 
of all the points, and $\mathbf{s}_i$ is the three coordinates of $i$th point. 
With the information of coordinates, different models, such as graphs \cite{c3} and octrees \cite{c5}, can be constructed to analyze the point clouds, for which we will discuss more in Section \ref{appli}.

\subsection{Tensor Basics}
Tensor is a high-dimensional generalization of matrix. A tensor can be 
interpreted as multi-dimensional arrays. The order of tensor is the 
number of indices to label the components of arrays \cite{c17}. 
For example, a scalar is 
a zeroth-order tensor; a vector is a first-order tensor; a matrix is a second-order tensor; and an $M$-dimensional array is an $M$th-order tensor \cite{c18}. 
In this work, an $M$th-order tensor is denoted by $\mathbf{A}\in\mathbb{R}^{I_1
	\times I_2\times \cdots \times I_M}$, whose entry in position $(i_1,i_2,\cdots,i_M)$ is labeled as $a_{i_1\cdots i_M}$. Here, $I_k$ is the dimension of $k$th order. 

Tensor outer product is a widely used operation to construct a 
higher-order tensor from lower-order tensors. The tensor outer product between an $M$th-order tensor $\mathbf{U}\in \mathbb{R}^{I_1\times I_2\times ...\times I_M }$ with entries $u_{i_1 ... i_M}$ and an $N$th-order tensor $\mathbf{V}\in \mathbb{R}^{J_1\times J_2\times ...\times J_N }$ with 
entries $v_{j_1 ... j_N}$ is denoted by 
\begin{equation}
\mathbf{W}=\mathbf{U} \circ \mathbf{V},
\end{equation}
where the result $\mathbf{W}\in \mathbb{R}^{I_1\times I_2\times ...\times I_M \times J_1 \times J_2 \times ... \times J_N}$ is an $(M+N)$th-order tensor with entries
\begin{equation}
w_{i_1 ... i_M j_1 ... j_N}= u_{i_1 ... i_M} \cdot v_{j_1 ... j_N}.
\end{equation}

\subsection{Hypergraph Signal Processing}
Hypergraph signal processing (HGSP) is a tensor-based framework  
\cite{c9}. In the HGSP framework, a hypergraph with $N$ nodes and longest 
hyperedge connecting $M$ nodes, is represented by an $M$-th order $N$-dimension representing tensor $\mathbf{A}=(a_{i_1i_2\cdots i_M})\in\mathbb{R}^{N^M}$. The representing tensor can be adjacency tensor or Laplacian tensor in different purposes \cite{c16}. In this paper, we refer the adjacency tensor as the representing tensor, in which each entry $a_{i_1i_2\cdots i_M}$ indicates whether nodes $\{\mathbf{v}_1,\mathbf{v}_2,\cdots,\mathbf{v}_M\}$ are connected. The computation of the edge weight can be found in \cite{c9}.

With the orthogonal-CP decomposition, the representing tensor can be decomposed 
via 
\begin{equation}\label{decom}
\mathbf{A}=\sum_{r=1}^{N}\lambda_r\cdot\underbrace{\mathbf{f}_r\circ...\circ \mathbf{f}_r}_{\text{$M$ times}},
\end{equation}
where $\mathbf{f}_r$'s are orthonormal basis called spectrum components and $\lambda_r$ are frequency coefficients related to the hypergraph frequency. All the spectrum components $\{\mathbf{f}_1,\cdots,\mathbf{f}_N\}$ construct the hypergraph spectral space. Each pair $(\mathbf{f}_r,\lambda_r)$ is called the spectral pair of the hypergraph.

Given an original signal $\mathbf{s}=[s_1\quad s_2\quad...\quad s_N]^{\mathrm{T}}$, the hypergraph signal is defined as the $(M-1)$ times tensor outer product of $\mathbf{s}$, i.e.,
\begin{equation}
\mathbf{s}^{[M-1]}=\underbrace{\mathbf{s\circ...\circ s}}_{\text{$M-1$ times}}.
\end{equation}

The hypergraph frequency is ordered by the total variation of the spectrum component, which is defined as
\begin{equation}
\mathbf{TV}(\mathbf{\mathbf{f}_r})=||\mathbf{f}_r-\frac{1}{\lambda_{max}}\mathbf{A}\mathbf{f}_r^{[M-1]}||_1,
\end{equation}
where $\mathbf{A}\mathbf{f}_r^{[M-1]}$ is the contraction between representing tensor $\mathbf{A}$ and the hypergraph signal.

A spectrum component with larger total variation is a higher-frequency component, which indicates a faster propagation over the given hypergraph. Moreover, a supporting matrix 
\begin{equation}\label{sup}
\mathbf{P_s}=\frac{1}{\lambda_{\max}}
\begin{bmatrix}
\mathbf{f}_1& \cdots& \mathbf{f}_N
\end{bmatrix}
\begin{bmatrix}
\lambda_1& & \\
&\ddots& \\
& &\lambda_N
\end{bmatrix}
\begin{bmatrix}
\mathbf{f}_1^{\mathrm{T}}\\
\vdots\\
\mathbf{f}_N^{\mathrm{T}}
\end{bmatrix},
\end{equation}
can be defined to capture the overall spectral information of the hypergraph.

Instead of reviewing many properties of HGSP here, other aspects such as hypergraph Fourier transform, hypergraph filter design and sampling theory can be found in \cite{c9}.

\section{Hypergraph Spectrum Estimation for Point Clouds}\label{h1}
To process the 3D point clouds, the first step is to construct an optimal 
hypergraph to model the point clouds. As we mentioned in the Section \ref{intro}, 
it is time-comsuming and inefficient
to first construct a hypergraph structure 
before tensor decomposition to obtain the hypergraph spectrum. 
Instead, we propose to directly estimate the hypergraph spectral pairs 
based on the observed data, and then recover the original representing tensor 
with Eq. (\ref{decom}). In this section, we 
first estimate the hypergraph spectrum components $\mathbf{f}_r$'s based on the hypergraph stationary process, and optimize the frequency coefficients $\lambda_r$'s based on the smoothness for original point clouds.

\subsection{Estimation of Hypergraph Spectrum Components}
In this part, we propose a method to estimate the hypergraph spectral 
components based on the hypergraph stationary process.
 
\subsubsection{Hypergraph Stationary Process}
Before providing details of the estimation, let us first introduce some new definitions and properties necessary for spectrum estimation.

Stationarity is a cornerstone property that facilities the analysis 
of random signals and observations in traditional signal processing \cite{c19}.
It has equal importance in graph and hypergraph signal processing. 
Based on graph shifting introduced in \cite{c22}, a definition of graph stationary
process proposed in \cite{c19} can analyze the properties of the 
different observations of nodes, or the random signals over the graphs. 
Furthermore, \cite{c123} introduces a method to estimate the graph spectrum space and graph diffusion for multiple observations based on the graph stationary process. Similarly, the hypergraph stationary process can be defined to 
estimate hypergraph spectrum. 


Now, let us introduce the definition of the hypergraph stationary process. In \cite{c9}, a polynomial hypergraph filter based on supporting matrix is defined as
\begin{equation}
\mathbf{s}'=\sum_{k=1}^{a}\alpha_k\mathbf{P}^k\mathbf{s},
\end{equation}
where $\mathbf{P}=\lambda_{max}\mathbf{P_s}$.

 Similarly, based on the supporting matrix, a $\tau$-step shifting operation is defined as
$\mathbf{P}_{\tau}=\mathbf{P}^{\tau}$. Then, similar to the definition of the stationary process in traditional digital signal processing and graph signal processing, a strict-sense stationary process in HGSP can be defined as follows.

\begin{myDef}(Strict-Sense Stationary Process)
	A stochastic signal $\mathbf{x}\in\mathbb{R}^N$ is strict-sense stationary over the hypergraph with $\mathbf{P}_\tau$ if and only if
	\begin{equation}
	\mathbf{x}\overset{d}{=}\mathbf{P}_{\tau}\mathbf{x}
	\end{equation}
	holds for any $\tau$.
\end{myDef}

Since the strict-sense stationary is hard to achieve and analyze in the real datasets, we introduce the weak-sense stationary process similar to traditional digital signal processing.

\begin{myDef}(Weak-Sense Stationary Process)
	A stochastic signal $\mathbf{x}\in\mathbb{R}^N$ is weak-sense stationary over the hypergraph with $\mathbf{P}_\tau$ if and only if
	\begin{equation}\label{mean}
	\mathbb{E}[\mathbf{x}]=\mathbb{E}[\mathbf{P}_{\tau}\mathbf{x}]
	\end{equation}
	and
	\begin{equation}\label{time}
	\mathbb{E}[(\mathbf{P}_{\tau_1}\mathbf{x})((\mathbf{P}^H)_{\tau_2}\mathbf{x})^H]=\mathbb{E}[(\mathbf{P}_{\tau_1+\tau}\mathbf{x})((\mathbf{P}^H)_{\tau_2-\tau}\mathbf{x})^H]
	\end{equation}
	hold for any $\tau$,
	where $\mathbb{E}(\cdot)$ refers to the mean of observations and $(\cdot)^H$ is the Hermitian transpose.
\end{myDef}

From the definition of the weak-sense stationary process (WSS), 
Eq. (\ref{mean}) implies that the mean function of the signal must be constant, 
which is the same condition as in traditional digital signal processing
(DSP) \cite{c23}. 
From the definition of supporting matrix, the $(i,j)$-th entry of $\mathbf{P}$ is the same as the $(j,i)$-th entry of $\mathbf{P}^H$, which indicates that $\mathbf{P}^H$ is the shifting in the opposite direction of $\mathbf{P}$. Then, the condition in Eq. (\ref{time}) indicates that the hypergraph covariance function $K_{\mathbf{xx}}({\tau_1},-\tau_2)=K_{\mathbf{xx}}({\tau_1}+\tau,\tau-\tau_2)=K_{\mathbf{xx}}({\tau_1}+\tau_2,0)$, which is also consistent with the definition in 
traditional DSP.

With the definition of the hypergraph stationary process, we have the following properties regarding the relationship between signals and hypergraph spectrum.

\begin{theorem}
	A stochastic signal $\mathbf{x}$ is WSS if and only if it has zero-mean and its covariance matrix has the same eigenvectors as the hypergraph spectrum basis, i.e., 
	\begin{equation}\label{s1}
	\mathbb{E}[\mathbf{x}]=\mathbf{0}
	\end{equation}
	and 
	\begin{equation} \label{s2}
	\mathbb{E}[\mathbf{x}\mathbf{x}^H]=\mathbf{V}\Sigma_\mathbf{x}\mathbf{V}^{H},
	\end{equation}
	where $\mathbf{V}=[\mathbf{f}_1,\mathbf{f}_2,\cdots,\mathbf{f}_N]\in\mathbb{R}^{N\times N}$ are the hypergraph spectrum.
	
\end{theorem}
\begin{proof}
	Since the hypergraph spectrum basis are orthonormal, we have $\mathbf{VV}^T=\mathbf{I}$.
Then, the $\tau$-step shifting based on supporting matrix can be calculated as
	\begin{align}
	\mathbf{P}_\tau&=\underbrace{\mathbf{V}\Lambda_\mathbf{P}\mathbf{V}^T\mathbf{V}\Lambda_\mathbf{P}\mathbf{V}^T\cdots\mathbf{V}\Lambda_\mathbf{P}\mathbf{V}^T}_{\tau\quad times}\\
	&=\mathbf{V}\Lambda_\mathbf{P}^\tau\mathbf{V}^T.\label{poly}
	\end{align}
	
	Now, the Eq. (\ref{mean}) can be written as 
	\begin{equation}
	\mathbb{E}[\mathbf{x}]=\mathbf{V}\Lambda_P^\tau\mathbf{V}^T\mathbb{E}[\mathbf{x}].
	\end{equation}
Since $\mathbf{V}\Lambda_P^\tau\mathbf{V}^T$ does not always equal to $\mathbf{I}$, Eq. (\ref{mean}) holds for arbitrary supporting matrix and $\tau$ if and only if $\mathbb{E}[\mathbf{x}]=\mathbf{0}$. 
	
Next we show the sufficiency and necessity of the condition in Eq. (\ref{s2}). 
	The condition in Eq. (\ref{time}) can be written as
	\begin{equation}
	\mathbf{P}_{\tau_1}\mathbb{E}[\mathbf{xx}^H]((\mathbf{P})^H_{\tau_2})^H=\mathbf{P}_{\tau_1+\tau}\mathbb{E}[\mathbf{xx}^H]((\mathbf{P})^H_{\tau_2-\tau})^H.
	\end{equation}
Considering Eq. (\ref{poly}) and the fact that
hypergraph spectrum is real \cite{c9}, Eq. (\ref{time}) is equivalent to
	\begin{equation}
	\mathbf{V}\Lambda_\mathbf{P}^{\tau_1}\mathbf{V}^H\mathbb{E}[\mathbf{xx}^H]\mathbf{V}\Lambda_\mathbf{P}^{\tau_2}\mathbf{V}^H=\mathbf{V}\Lambda_\mathbf{P}^{\tau_1+\tau}\mathbf{V}^H\mathbb{E}[\mathbf{xx}^H]\mathbf{V}\Lambda_\mathbf{P}^{\tau_2-\tau} \mathbf{V}^H,
	\end{equation}
	which can be written as
	\begin{equation} \label{eig}
	(\mathbf{V}^H\mathbb{E}[\mathbf{xx}^H]\mathbf{V})\Lambda_\mathbf{P}^{\tau}=\Lambda_\mathbf{P}^{\tau}(\mathbf{V}^H\mathbb{E}[\mathbf{xx}^H]\mathbf{V}).
	\end{equation}
	If Eq. (\ref{eig}) holds for arbitrary $\mathbf{P}$, 	$(\mathbf{V}^H\mathbb{E}[\mathbf{xx}^H]\mathbf{V})$ should be diagonal, 
which indicates $\mathbb{E}[\mathbf{x}\mathbf{x}^H]=\mathbf{V}\Sigma_\mathbf{x}\mathbf{V}^{H}$. 
Thus, the sufficiency of the condition is proved. 
	
Similarly, we can apply Eq. (\ref{s2}) on both sides of Eq. ($\ref{time}$), 
we can establish the necessity of the condition in Eq. (\ref{s2}). 
\end{proof}

This theorem can be used to estimate the hypergraph spectrum, given multiple observations of several signal points.

\subsubsection{Estimation of Spectrum Components for Point Clouds}
Now, we can use the property of stationary process to estimate the hypergraph spectrum of point clouds.
The three coordinates of a point can be interpreted as three observations of the point from different angles, which describe the underlying multilateral relationship. Thus, we can assume that the point cloud signals follow the stationary process over the estimated underlying hypergraph structure. If the point cloud signals $\mathbf{s}$ follow the hypergraph stationarity, it should satisfy Eq. (\ref{s1}) and Eq. (\ref{s2}). Thus, a spectrum estimation method can be
based on hypergraph staionarity. 
The details of the algorithm is described as follows.

\begin{algorithm}[htbp]
	\begin{algorithmic}[1] 
		\caption{Estimation of Hypergraph Spectrum}\label{basisestimation}
		\STATE {\bf{Input}}: Point cloud dataset $\mathbf{s}=[\mathbf{X_1\quad X_2\quad X_3}]\in\mathbb{R}^{N\times 3}$.
		\STATE Calculate the mean of each row in $\mathbf{s}$, i.e., \\$\mathbf{\overline s}=(\mathbf{X_1+X_2+X_3})/3$;
		\STATE Normalize the original point cloud data as zero-mean in each row, i.e.,
		$\mathbf{s}'=[\mathbf{X_1-{\overline s},X_2-{\overline s},X_3-{\overline s}}]$;
		\STATE Calculate the eigenvectors $\{\mathbf{f}_1,\cdots,\mathbf{f}_N\}$ for $R_{\mathbf{s}'}=\mathbf{s'}(\mathbf{s'}^T)$;
		\STATE {\bf{Output}}: Hypergraph spectrum $\mathbf{V}=[\mathbf{f}_1,\cdots,\mathbf{f}_N]$.
	\end{algorithmic}
\end{algorithm}

With \textit{Theorem 1}, we can directly
obtain an estimation of the hypergraph spectrum based on the 
hypergraph stationarity. Note that, here, we assume all the observations 
are from a clean point cloud without noise.
The case of noisy point clouds will be discussed later in Section \ref{h2}.

\subsection{Estimation of Frequency Coefficients}\label{es}

Next, we discuss how we estimate the hypergraph frequency coefficients with the spectrum components based on the hypergraph smoothness.

In real applications, the large-scale networks are usually sparse, which makes it meaningful to infer that most entries of the hypergraph representing tensor for real datasets are zero \cite{c24}. In addition, the smoothness of signals is a widely-used assumption when estimating the underlying structure of graphs and hypergraphs \cite{c25}. Thus, the estimation of the hypergraph representing 
tensor with known spectrum components for a given dataset $\mathbf{s}$ can be generally formulated as
\begin{align} 
&\min_{\mathbf{\boldsymbol{\lambda}}}\quad \alpha \mbox{Smooth}
(\mathbf{s,\boldsymbol{\lambda}},\mathbf{f}_r)+\beta||\mathbf{A}||_T^2\label{e1}\\
s.t.\quad &\mathbf{A}=\sum_{r=1}^{N}\lambda_r\cdot\underbrace{\mathbf{f}_r\circ...\circ \mathbf{f}_r}_{\text{M times}}. \label{dec}\\
&\quad\mathbf{A}\in \mathcal{A}.\label{cs1}\\
&||\mathbf{A}||_T=\sqrt{\sum_{i_1,i_2,\cdots, i_M=1}^N a_{i_1i_2\cdots i_M}^2}.\label{t_norm}
\end{align}

The constraint set $\mathcal{A}$ in (\ref{cs1}) 
includes the prior information of the representing tensor. 
For example, if the representing
tensor is the adjacency tensor, its entries should be non-negative. 
In the constraint of $(\ref{t_norm})$, $||\mathbf{A}||_T$ is the tensor norm
which controls the sparsity of the hypergraph structure. The smoothness function Smooth$(\mathbf{s,\boldsymbol{\lambda}},\mathbf{f}_r)$ can be designed for specific problems. Typical functions can be hypergraph Laplacian regularization, label ranking, and total variation \cite{c9}. For convenience, we use the quadratic-form total variation based on the supporting matrix to describe the hypergraph smoothness, i.e.,
\begin{equation}
\mathbf{TV}(\mathbf{s})=||\mathbf{s}-(1/\lambda_{max})\mathbf{P}\mathbf{s}||^2_2.
\end{equation}
This form of smoothness function suggested in \cite{c9} can capture the differences between one node and its neighbors over hypergraph. Since the signals are smooth over the estimated hypergraph, 
observations are also smooth. Thus, the final smoothness function for point cloud $\mathbf{s}=[\mathbf{X_1\quad X_2\quad X_3}]$ is 
\begin{align}
\mbox{Smooth}(\mathbf{s,\boldsymbol{\lambda}},\mathbf{f}_r)&=\sum_{i=1}^3||\mathbf{X}_i-\mathbf{P_sX}_i||^2_2\nonumber\\
&=\sum_{i=1}^3||\mathbf{X}_i-\sum_r \sigma_r(\mathbf{f}_r^T\mathbf{X}_i)\mathbf{f}_r||^2_2\nonumber\\
&=\sum_{i=1}^3||\mathbf{X}_i-\mathbf{W}_i\boldsymbol{\sigma}||^2_2,\label{sm1}
\end{align}
where $\mathbf{W}_i=[(\mathbf{f}_1^T\mathbf{X}_i)\mathbf{f}_1\quad(\mathbf{f}_2^T\mathbf{X}_i)\mathbf{f}_2\quad\cdots\quad(\mathbf{f}_N^T\mathbf{X}_i)\mathbf{f}_N]$, $\sigma_r=\lambda_r/\lambda_{max}$ and $\boldsymbol{\sigma}=[\sigma_1\cdots\sigma_N]$.

Moreover, the tensor norm of a given hypergraph has the following property with the frequency coefficients.
\begin{theorem}
	Given a representing tensor $\mathbf{A}=\sum_{r=1}^{N}\lambda_r\cdot\underbrace{\mathbf{f}_r\circ...\circ \mathbf{f}_r}_{\text{M times}}$, the tensor norm $||\mathbf{A}||^2_T=\sum_{i_1,i_2,\cdots,i_M=1}^{N}a_{i_1i_2\cdots i_M}^2$ can be written in the form of frequency coefficients as
	\begin{equation}
	||\mathbf{A}||^2_T=\sum_{r=1}^N \lambda_r^2=\boldsymbol{\lambda}^T\boldsymbol{\lambda},
	\end{equation}
where $\boldsymbol\lambda=[\lambda_1 \quad\lambda_2 \quad\cdots\quad\lambda_N]$.
\end{theorem}
\begin{proof}
	Since $\mathbf{A}=\sum_{r=1}^{N}\lambda_r\cdot\underbrace{\mathbf{f}_r\circ...\circ \mathbf{f}_r}_{\text{M times}}$, we have
	\begin{equation}
	a_{i_1i_2\cdots i_M}=\sum_{r=1}^N \lambda_rf_{r,i_1}f_{r,i_2}\cdots f_{r,i_M},
	\end{equation}
where $f_{r,i}$ is the $i$th element of $\mathbf{f}_r$.
	Then, the tensor norm is 
	\begin{align}
	||\mathbf{A}||^2_T
	&=\sum_{i_1,i_2,\cdots,i_M}(\sum_{r=1}^N \lambda_rf_{r,i_1}f_{r,i_2}\cdots f_{r,i_M})^2\nonumber\\
	&=\sum_{i_1,i_2,\cdots,i_M}(\sum_{r=1}^N \lambda_rf_{r,i_1}\cdots f_{r,i_M})(\sum_{t=1}^N \lambda_tf_{t,i_1}\cdots f_{t,i_M})\nonumber\\
	&=\sum_{i_1,i_2,\cdots,i_M}\sum_{r,t}\lambda_r\lambda_tf_{r,i_1}\cdots f_{r,i_M} f_{t,i_1}\cdots f_{t,i_M}\nonumber\\	
	&=\sum_{r,t}\lambda_r\lambda_t\sum_{i_1,i_2,\cdots,i_M=1}^N(f_{r,i_1}f_{t,i_1})\cdots(f_{r,i_M}f_{t,i_M})\nonumber\\
	&=\sum_{r,t}\lambda_r\lambda_t(\mathbf{f}_r^T\mathbf{f}_t)^M.
	\end{align}
Since $\mathbf{f}_r$ is orthogonal, $\mathbf{f}_r^T\mathbf{f}_t=1$ holds if $r=t$; otherwise, $\mathbf{f}_r^T\mathbf{f}_t=0$. Thus, we obtain $||\mathbf{A}||^2_T=\sum_{r=1}^N \lambda_r^2$.
\end{proof}
This property can help us build a connection from the tensor norm to the frequency coefficients directly.

Now, if we consider the representing tensor as the adjacency tensor and each hyperedge consists of three nodes since at least three nodes are required
to construct a surface, we optimize the normalized 
frequency coefficients  $\boldsymbol{\sigma}=\frac{1}{\lambda_{max}}\boldsymbol{\lambda}=[\sigma_1\quad 
\sigma_2 \quad \cdots\quad \sigma_N]^T$ via
\begin{align}\label{target}
& \min_{\boldsymbol{\sigma}}\quad \alpha\sum_{i=1}^3||\mathbf{X}_i-\mathbf{W}_i\boldsymbol{\sigma}||^2_2+\beta{ \boldsymbol\sigma^T\boldsymbol\sigma}\\
\hspace{-3mm}s.\; t. \; &\;\;0\leq \sigma_r\leq \max_i \sigma_{i}=1,\label{non}\\
&\sum_{r=1}^N \sigma_r f_{r,i_1}f_{r,i_2}f_{r,i_3}\geq 0, \quad i_1,i_2,i_3=1,2,\cdots,N.\label{adj}
\end{align}

The constraint (\ref{adj}) limits the estimated representing tensor as the adjacency tensor. The constraint (\ref{non}) is the nonnegative constraint on weight and the factor \cite{c26}. Clearly, the optimization is non-convex with the constraint 
$\max_i \sigma_{i}=1$. However, if the position of the maximal frequency is known, 
the optimization problem can be solved by tools such as cvx \cite{c27,c28}. 
Thus, we can develop the following algorithm to estimate the frequency coefficients.

\begin{algorithm}[htbp]
	\begin{algorithmic}[1] 
		\caption{Estimation of Frequency Coefficient}\label{freestimation}
		\STATE {\bf{Input}}: Point cloud dataset $\mathbf{s}=[\mathbf{X_1,X_2,X_3}]\in\mathbb{R}^{N\times 3}$, hypergraph spectrum $\mathbf{V}=[\mathbf{f}_1,\cdots,\mathbf{f}_N]$.
		\STATE {\bf{for}} i=1,2,...,iter {\bf{do}}:
		\STATE 	\quad Set $\sigma_i=1$ as the maximal normalized eigenvalue.
		\STATE 	\quad Solve the optimization problem in Eq. (\ref{target}).
		\STATE {\bf{end for}}
		\STATE Find the optimal $i$ to minimize the target function.
		\STATE The optimal coefficients \boldsymbol{$\sigma$} is the solution of Eq. (\ref{target}) correlated to the optimal $i$.
		\STATE {\bf{Output}}: Frequency coefficients {$\boldsymbol\sigma$}.
	\end{algorithmic}
\end{algorithm}

Note that, since we consider clean point cloud
without noise, we usually set parameter $\alpha\ll\beta$. 
Then, from the estimated spectrum pair $(\mathbf{f}_r,\sigma_r)$ under normalization, 
we can recover the original adjacency tensor as Eq. (\ref{dec}). 
Hence, the hypergraph construction process
for a clean point cloud can be summarized as Fig. \ref{fram1}. 
The recovery of original adjacency tensor is not always necessary in practical applications since 
storing the representing tensor is less efficient than storing
the spectrum pairs.
\begin{figure}[t]
	\centering
	\includegraphics[width=3in]{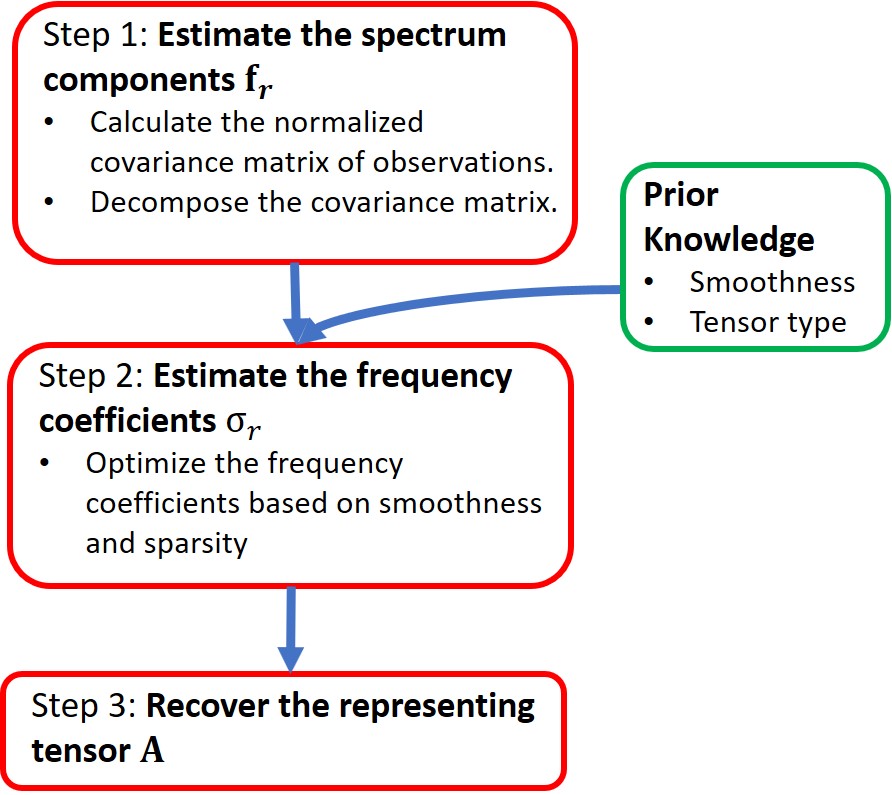}
	\caption{Estimation of Hypergraph Spectral Pairs for Original Point Clouds}
	\label{fram1}
\end{figure}

\section{Joint Spectrum Estimation and Denoising}\label{h2}
In practical 3D imaging, perturbations such as noises and outliers
often exist when generate a point cloud of an unknown object. 
These noises may significantly affect the performance of point cloud processing 
since many existing algorithms require quality datasets~\cite{c27}. 
Thus, denoising remains a vital issue in practical point cloud applications.

Usually, to denoise point sets with sharp features is difficult, 
especially when the noise is large, as such features are hard to
distinguish from noise effect. 
Generally, smoothness-based methods are common. In \cite{c30}, a method 
based on $L_0$ norm of differences between $k$-nearest neighbors 
is introduced. In \cite{c31}, Laplacian regularization is used to 
describe smoothness and to denoise noisy point sets. Other works, 
such as \cite{c8,c29}, minimize the total variation over graphs to 
denoise the point sets. 
Although smoothness-based methods have achieved notable 
successes, how to interpret and define an effective smoothness function 
for a general point set remains open.
Furthermore, for graph-based smoothness methods, the construction of graph model 
remains a critical problem, since traditional methods based on distance 
suffers from the imprecise location measurement. To this end, a 
more general definition of smoothness and a more efficient 
denoising method for arbitrary point clouds are highly desirable.

In this section, we introduce a joint method to simultaneously estimate the hypergraph structure and denoise noisy point clouds. 
In Section \ref{h1}, we already introduce an estimation method of spectral pair $(\mathbf{f}_r,\sigma_r)$ for clean point clouds. A similar construction process can be developed for the noisy point clouds. 
As the estimation of spectrum components only depends on the observed data, we need to denoise the noisy observations while optimizing the frequency coefficients. 
As already discussed, the problem of denoising a signal on a hypergraph can 
be written as a convex minimization problem 
with the constraints that denoised signals should be smooth 
over the hypergraph. Accordingly, the general process of hypergraph denoising and 
estimation can be summarized as the following steps:
\begin{itemize}
	\item Step 1: Estimate the approximated hypergraph spectrum components from the observed noisy point clouds;
	\item Step 2: Jointly estimate frequency coefficients and denoise the noisy observations;
	\item Step 3: Update the noisy observations as denoised data and 
	repeat Step 1 until enough iterations.
\end{itemize}

To estimate hypergraph spectral components of noisy data, 
the process is the same as Algorithm 1 based on hypergraph stationary process. 
To jointly estimate the frequency coefficients to recover the original underlying structure and to denoise the noisy point clouds, 
we propose the following objective. Given $N$ noisy points 
$\mathbf{s}=[\mathbf{X_1\quad X_2\quad X_3}]$, the joint estimation task
can be formulated as

\begin{align}
	\min_{\boldsymbol{\sigma},\mathbf{Y}}\quad& \sum_{i=1}^3[||\mathbf{X}_i-\mathbf{Y}_i||^2_2+\alpha ||\mathbf{X}_i-\mathbf{W}_i\boldsymbol{\sigma}||^2_2]+\beta ||\mathbf{A}||_2^2\label{esfram}\\
	s.t.\quad &\mathbf{A}=\sum_{r=1}^{N}\lambda_r\cdot\underbrace{\mathbf{f}_r\circ...\circ \mathbf{f}_r}_{\text{M times}} \in \mathcal{A},\nonumber\\
 &\;\;0\leq \sigma_r\leq \max_i \sigma_{i}=1, \nonumber\\
	&\mathbf{W}_i=[(\mathbf{f}_1^T\mathbf{X}_i)\mathbf{f}_1\quad(\mathbf{f}_2^T\mathbf{X}_i)\mathbf{f}_2\quad\cdots\quad(\mathbf{f}_N^T\mathbf{X}_i)\mathbf{f}_N].
	\nonumber
\end{align}

The resulting $\mathbf{Y=[Y_1\quad Y_2 \quad Y_3]}$ is the 
denoised point clouds, and $(\alpha, \beta)$ are two positive regularization parameters. The first part in Eq. (\ref{esfram}) lets the denoised point cloud maintain the observed structural features. The second part is the smoothness function derived from Eq. (\ref{sm1}) which adjusts positions of noisy points. 
The third part is the tensor norm regularization to 
control hypergraph sparsity.

The optimization problem of Eq. (\ref{esfram}) is not convex in $\mathbf{Y}$ and $\boldsymbol{\sigma}$. Therefore, similar to \cite{c25}, we split the problem into
two subproblems. For each subproblem, we fix one variable set
to solve the other one. Upon convergence, 
the solution corresponds to a local minimum and not necessarily 
a global minimum. 

We first initialize $\mathbf{Y}$ as the observed signals $\mathbf{X}$ and solve the following problem similar to that in Section \ref{h1}.
  \begin{equation}\label{target1}
\min_{\boldsymbol{\sigma}} \alpha\sum_{i=1}^3||\mathbf{X}_i-\mathbf{W}_i\boldsymbol{\sigma}||^2_2+\beta{ \boldsymbol\sigma^T\boldsymbol\sigma}
\end{equation}
\begin{align}
s.t.\quad
&\;\;0\leq \sigma_r\leq \max_i \sigma_{i}=1, \nonumber\\
&\sum_{r=1}^N \sigma_r f_{r,i_1}f_{r,i_2}f_{r,i_3}\geq 0, \quad i_1,i_2,i_3=1,2,\cdots,N.\nonumber
\end{align}
This problem can be solved similarly to the solution
of clean point cloud with Algorithm 2.

Once the estimated frequency coefficients are found, 
we solve the subproblem of point cloud denoising 
\begin{align}\label{ss2}
	\min_\mathbf{Y} \sum_{i=1}^3[||\mathbf{X}_i-\mathbf{Y}_i||^2_2+\alpha ||\mathbf{X}_i-\mathbf{W}_i\boldsymbol{\sigma}||^2_2],
\end{align}
whose close-form solution for each coordinate is 
\begin{equation}\label{ss3}
	\mathbf{Y}_i=[\mathbf{I}+\alpha(\mathbf{I-P_s})^T(\mathbf{I-P_s})]^{-1}\mathbf{X}_i.
\end{equation}

Note that $\mathbf{P_s}$ is the supporting matrix. We then update the frequency components based on the denoised point clouds, and repeatedly carry out
Step 1 to Step 3 until getting the final solution. 
In practice, we generally observe the convergence within only a few iterations.
The complete algorithm is summarized in Algorithm 3 as shown in Fig. \ref{fram2}.
Unlike for clear point clouds, we emphasize more on the smoothness of signals 
over the hypergraph.
The parameter $\alpha$ can be set larger than used when dealing with clean
point clouds.
\begin{algorithm}[htbp]
	\begin{algorithmic}[1] 
		\caption{Joint Hypergraph Estimation and Point Cloud Denoising}\label{jdenoise}
		\STATE {\bf{Input}}: Noisy observations of point clouds $\mathbf{s}=[\mathbf{X_1,X_2,X_3}]\in\mathbb{R}^{N\times 3}$.
		\STATE {\bf{Initialization}}: Calculate the spectrum components $\mathbf{f}_r$'s from the observed point cloud $\mathbf{s}$ as Algorithm 1.
		\STATE {\bf{for}} i=1,2,...,iter {\bf{do}}:
		\STATE 	\quad Find the optimal $\boldsymbol{\sigma}$ for the first subproblem in Eq. (\ref{target1}) with Algorithm 2.
		\STATE 	\quad Solve the optimization problem in Eq. (\ref{ss2}) with $\mathbf{Y}$ in Eq. (\ref{ss3}).
		\STATE 	\quad Update the observed signals as $\mathbf{Y}$ and recalculate the spectrum components $\mathbf{f}_r$'s.
		\STATE {\bf{end for}}
		\STATE {\bf{Output}}: Spectral pairs $(\mathbf{f}_r,\sigma_r)$'s, denoised point clouds $\mathbf{Y}$.
	\end{algorithmic}
\end{algorithm}
\begin{figure}[htbp]
	\centering
	\includegraphics[width=3in]{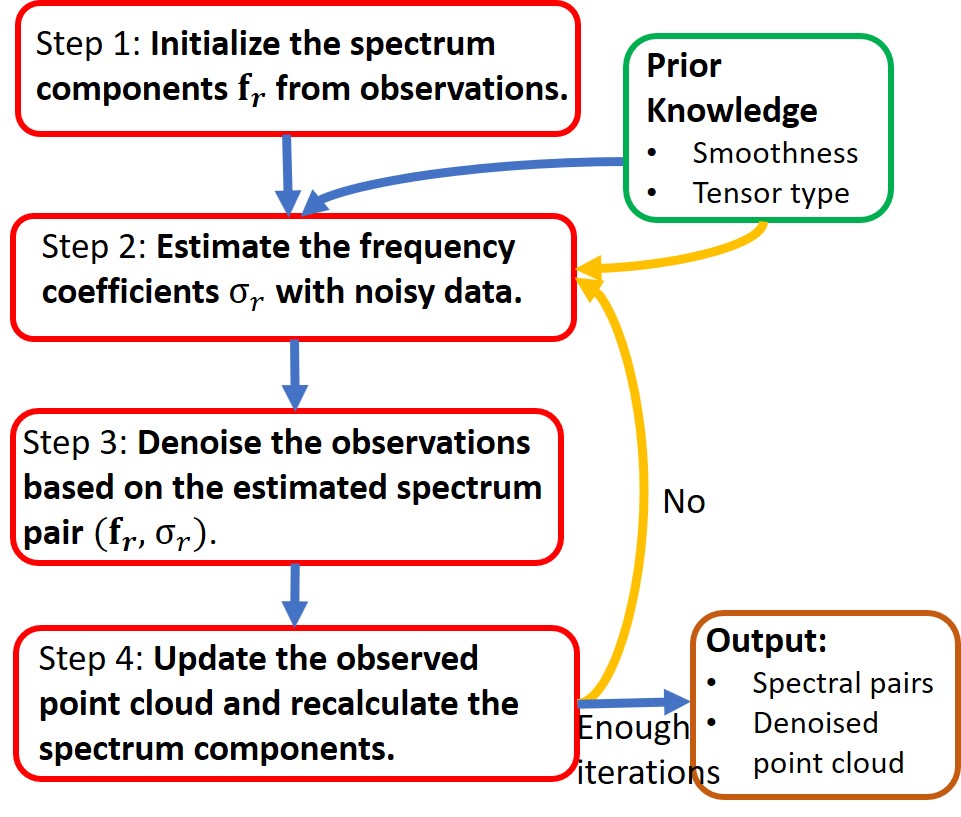}
	\caption{Joint Hypergraph Estimation and Denoising for Noisy Point Cloud.}
	\label{fram2}
\end{figure}

\section{Application Examples}\label{appli}
In this section, we examine two application examples to test the 
efficacy of the proposed method in estimating hypergraph structure for both 
clear and noisy point clouds.

\subsection{Sampling}
Sampling is an important operation to facilitate
analysis of very large point clouds. In this part, we consider 
different sampling strategies depending on different kinds of applications. 
Some interesting connections are found from the hypergraph frequency and 
point cloud features.

\subsubsection{Resampling using Harr-like Highpass Filtering}
Filtering helps extract select features of a given dataset. 
In some applications such as boundary detection, 
accurate extraction of shape features of point clouds is important. 
Thus, an efficient sampling should retain the features of the original point cloud. 
In our estimation of hypergraph structure, smoothness is a significant feature
to model point clouds. Ideally,  smoothness over the original surface of a point cloud
should correspond to smoothness over its hypergraph model.  Therefore, we can also
design a Harr-like high-pass filter to extract sharp features over the surfaces.

Let $\mathbf{I}$ be an identity matrix of appropriate size.  
Similar to that in GSP \cite{c3}, a Haar-like high-pass filter is designed as
\begin{align}
\mathbf{H}
&=\mathbf{I-P_s}\\
&=\mathbf{V}
\begin{bmatrix}
1-\sigma_1&0&\cdots&0\\
0&1-\sigma_2&\cdots&0\\
\vdots&\vdots&\ddots&\vdots\\
0&0&\cdots&1-\sigma_N
\end{bmatrix}\mathbf{V}^T.
\end{align}
The filtered signal is
\begin{equation}
(\mathbf{Hs})_i=\mathbf{s}_i-\sum_j {P_s}_{(ij)}\mathbf{s}_j,
\end{equation}
which reflects the differences between nodes and their 
neighbors over the hypergraph. Note that, the frequency 
coefficients together with their corresponding spectral components are ordered decreasingly here, i.e., $\sigma_{i}\geq\sigma_{i+1}$. From the definition of total variation, more smoothness corresponds to larger total variation. 
Thus, we can extract the sharp features over the point clouds by sampling the nodes with large value of $||\mathbf{s}_i-\sum_j {P_s}_{(ij)}\mathbf{s}_j||^2_2$.

To test this application,
we estimate the spectral pairs for clean point clouds and filter the signals 
over several synthetic datasets. We randomly generate multiple points 
over the surfaces of basic graphics shown as Fig. \ref{ori}, 
and sample the point clouds using the high-pass filter (HPF) given in Fig. \ref{sam}. 
From the test results, we can see that the sampled points of the surfaces 
in Fig. \ref{sur2} mainly congregate near the corners and edges, which are the sharp parts of the point clouds. 
In addition, the sampled nodes for a cube shape
are also crowded near edges and corners. On the other hand, 
the sampled nodes of a cylinder are mostly at the boundaries of
the cylinder. Our test 
results show that the Harr-like HPF can extract sharp features from 
point cloud surfaces, which correspond to the least
smooth parts of the estimated hypergraph. 
Moreover, since the total variation measures the order of frequency, 
sharp features over the point cloud correspond to high frequency components. 
Thus, the hypergraph model and the estimated spectral pairs 
are efficient when extracting features of 3D point clouds.

\begin{figure*}[htbp]
	\centering
	\subfigure[Original Surfaces with 6000 Points.]{
		\label{sur1}		
		\centering
		\includegraphics[height=1.7in]{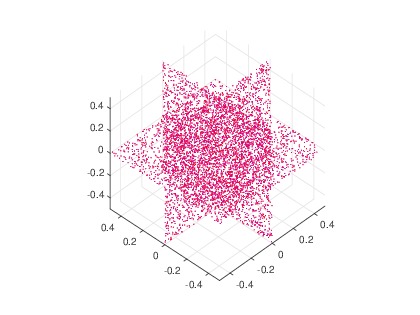}
	}
\hfill
\subfigure[Original Cylinder with 6000 Points.]{
	\label{cyl1}		
	\centering
	\includegraphics[height=1.7in]{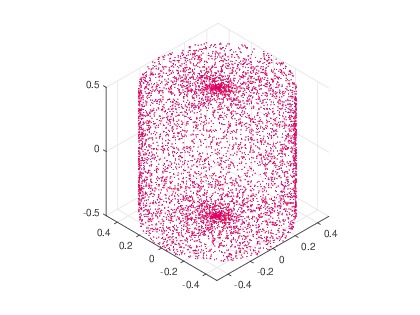}
}
\hfill
\subfigure[Original Cube with 5000 Points.]{
	\label{cub1}
	\centering
	\includegraphics[height=1.7in]{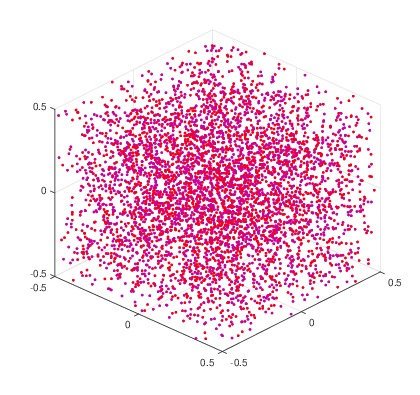}
}
	\centering
\caption{Original Point Clouds.}
\label{ori}
\end{figure*}
\begin{figure*}[htbp]
	\subfigure[Sampled Surfaces with 800 Points.]{
		\label{sur2}
		\centering
		\includegraphics[height=1.7in]{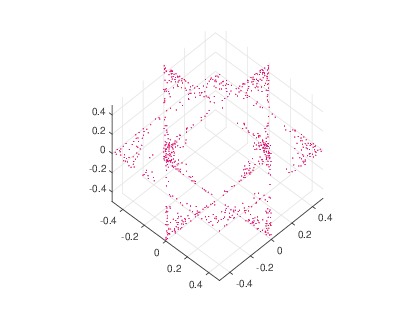}
	}%
\hfill
	\subfigure[Sampled Cylinder with 600 Points.]{
		\label{cyl2}		
		\centering
		\includegraphics[height=1.7in]{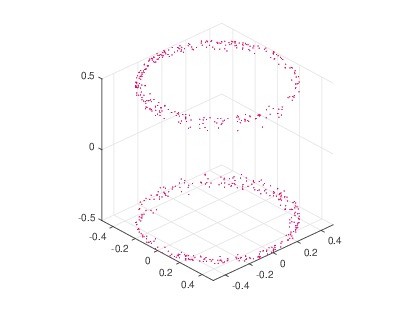}
	}
\hfill	
	\subfigure[Sampled Cube with 500 Points.]{
		\label{cub2}		
		\centering
		\includegraphics[height=1.7in]{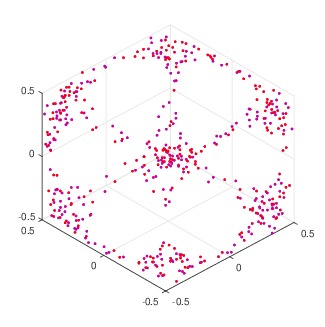}
	}
	\centering
	\caption{Sampled Point Clouds.}
	\label{sam}
\end{figure*}

\begin{figure*}[htbp]
	\subfigure[Sampled Surfaces with 800 Points.]{
	\label{sur2}
	\centering
	\includegraphics[height=1.5in]{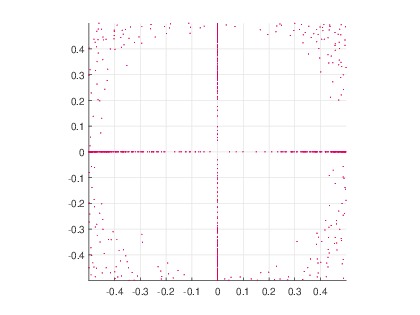}
}%
\hfill
\subfigure[Sampled Cylinder with 600 Points.]{
	\label{cyl2}		
	\centering
	\includegraphics[height=1.5in]{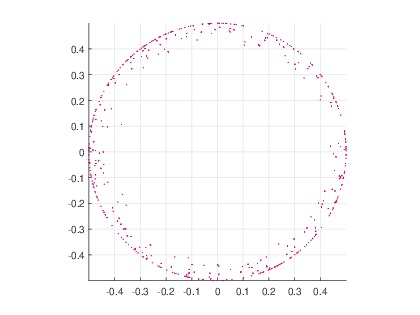}
}
\hfill	
\subfigure[Sampled Cube with 500 Points.]{
	\label{cub2}		
	\centering
	\includegraphics[height=1.5in]{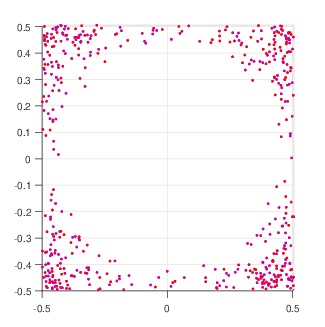}
}
	\caption{View from the Top of Sampled Point Clouds.}
	\label{sam5}
\end{figure*}

\subsubsection{Down-Sampling with Hypergraph Fourier Transform}
Projecting signals into a suitable orthonormal basis is a widely-used sampling method \cite{c32}. The work of \cite{c9} develops a sampling theory based on hypergraph signal processing as follows:
\begin{itemize}
	\item Step 1: Order the spectrum components from low frequency to high frequency based on their total variations.
	\item Step 2: Implement hypergraph Fourier tranform as
	\begin{equation}
		\mathcal{F}(\mathbf{s})=[(\mathbf{f}_1^T\mathbf{s})^{M-1}\quad (\mathbf{f}_2^T\mathbf{s})^{M-1}\quad \cdots\quad (\mathbf{f}_N^T\mathbf{s})^{M-1}]^T.
	\end{equation} 
	\item Step 3: Use $C$ transformed signal components 
	in the hypergraph frequency domain to represent $N$ signals in the original vertex domain.
\end{itemize}

More specifically, for a $K$-bandlimitted hypergraph signal, a perfect recovery is available with $K$ samples in hypergraph frequency domain.  Similarly, we can sample the point clouds based on the hypergraph Fourier transform. To test the performance of the sampled signals, we implement hypergraph Fourier transform (HGFT) on each coordinates of the point clouds, i.e., $\mathcal{F}(\mathbf{X}_i)$ for all $i$. Then, we take
the first $C$ transformed signals in all coordinates. Finally, 
we implement the inverse hypergraph Fourier transform (iHGFT) to obtain the sampled shapes of the original point clouds. Note that, perfect recovery happens 
with $C$ samples, if $(\mathcal{F}(\mathbf{X}_i))_{j+C}=0$ for  $i,j\in {\cal Z}^+$.
\begin{figure*}[t]
	\subfigure[Cat with 3400 points.]{
		\label{cat}
		\centering
		\includegraphics[height=1.4in]{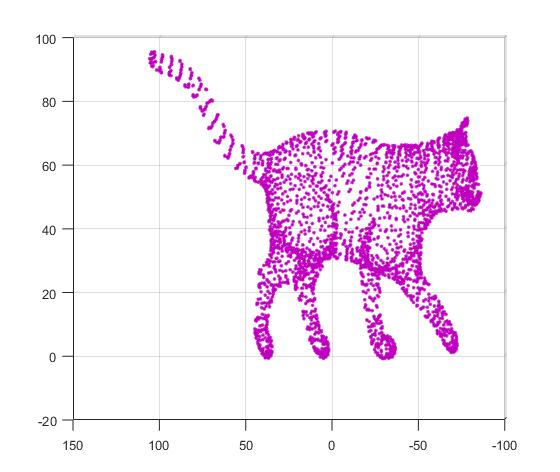}
	}%
	\hfill
	\subfigure[Wolf with 3400 points.]{
		\label{wolf}		
		\centering
		\includegraphics[height=1.4in]{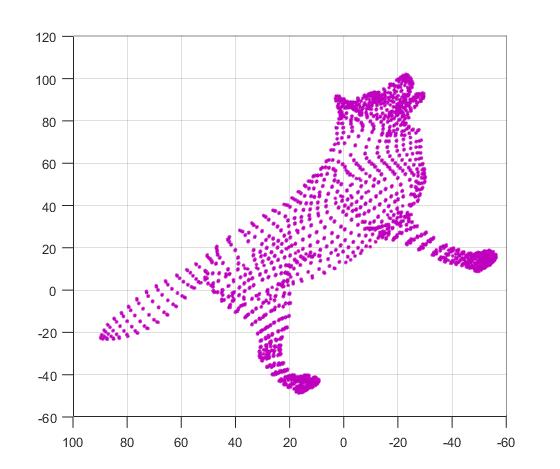}
	}	
	\hfill
	\subfigure[Horse with 3400 points.]{
		\label{hor}		
		\centering
		\includegraphics[height=1.4in]{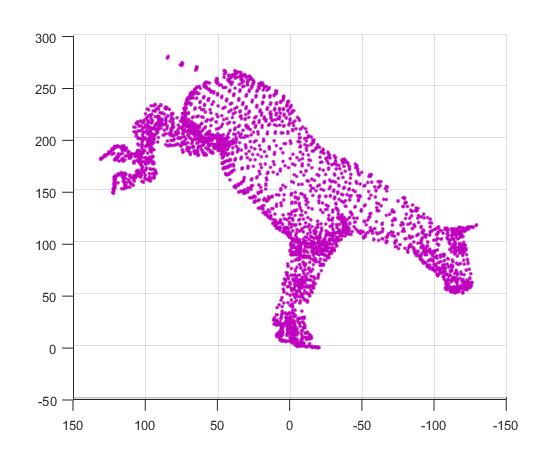}
	}
	\centering
	\caption{Test Datasets of Sampling.}
	\label{sam2}
\end{figure*}

\begin{figure*}[t]
	\subfigure[MSE for cat dataset.]{
		\label{cat1}
		\centering
		\includegraphics[height=1.8in]{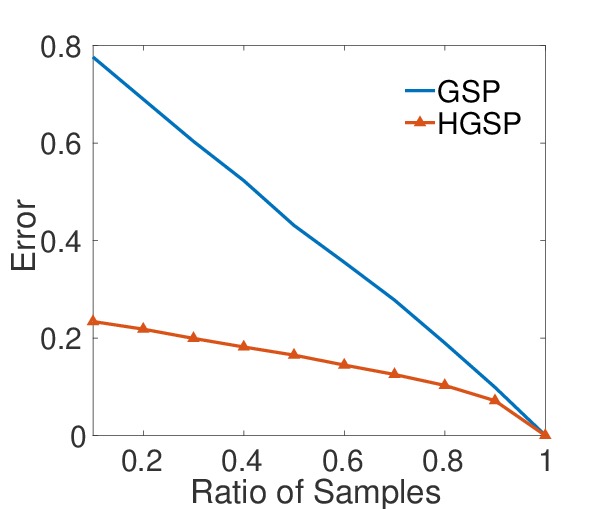}
	}%
	\hfill
	\subfigure[MSE for wolf dataset.]{
		\label{wolf1}		
		\centering
		\includegraphics[height=1.8in]{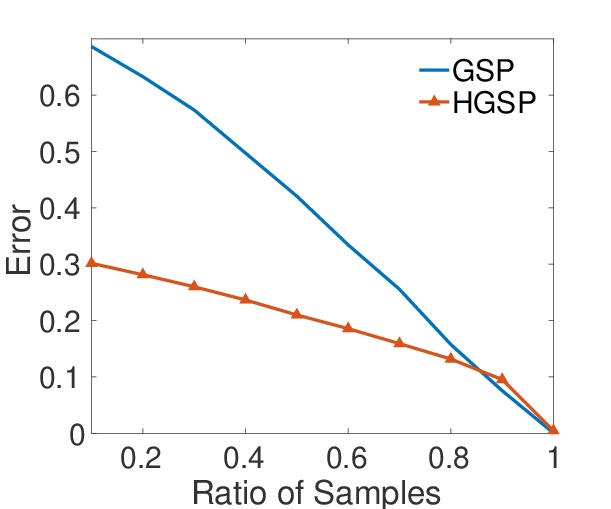}
	}	
	\hfill
	\subfigure[MSE for horse dataset.]{
		\label{hor1}		
		\centering
		\includegraphics[height=1.8in]{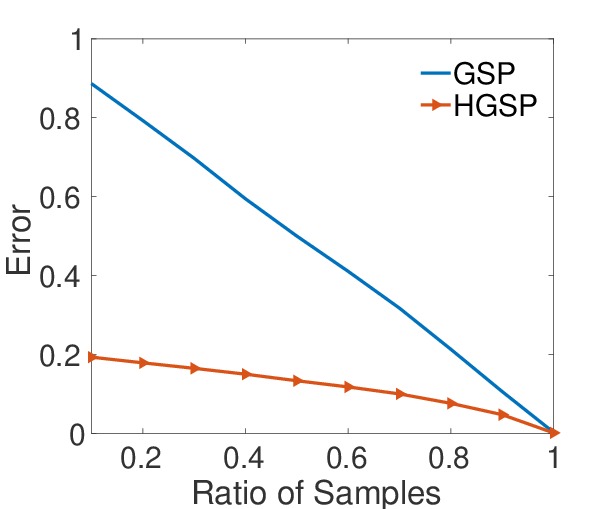}
	}
	\centering
	\caption{Error between Recovered Data and Original Data.}
	\label{sam3}
\end{figure*}
 We test the recovered point clouds for animal point datasets \cite{c33,c34,c35,c36} with the GSP-based methods. For the GSP-based method, we construct the a graph adjacency matrix $\mathbf{W}$ with Guassian model, i.e.,
\begin{equation}\label{adjj}
W_{ij}=\left\{
\begin{aligned}
\exp\left(-\frac{||\mathbf{s}_i-\mathbf{s}_j||^2_2}{\delta^2}\right),&\quad||\mathbf{s}_i-\mathbf{s}_j||^2_2\leq t;\\
0,&\quad \mbox{otherwise},
\end{aligned}
\right.
\end{equation}
where $\mathbf{s}_i$ is the coordinates of the $i$th node.
Then, we sample the point clouds using 
the signals after the graph Fourier transform (GFT). 

The test point cloud is shown as Fig. \ref{sam2}. 
We first compare the mean squared error (MSE)
between the recovered point clouds and original point clouds shown 
as Fig. \ref{sam3}. 
From the experimental results, we can see that the HGSP-based method 
has smaller error than the GSP downsampling method, clearly 
indicating hypergraph to be a better model. 
However, sometimes, MSE alone cannot tell the true story in terms of the performance 
for the recovered point clouds. 
To explore more, we compare the recovered point clouds directly in 
Fig. \ref{sam4}. From the experimental results, we can see that HGSP-based 
method captures the overall structure of the point clouds with very few samples, 
whereas the GSP-based method requires more samples to get sufficient details.
The MSE of GSP mainly stems from some outliers when taking more than 90 percent 
of the samples. 
The experiments show that HGSP-based method is a better tool for applications 
which need to recover an overall shape of point clouds from limited data
storage. Our test shows hypergraph to be a suitable model for point clouds, and the estimated hypergraph spectral pairs capture the point cloud characteristics 
very well.

\subsection{Denoising}
\begin{figure}[t]
	\subfigure[Bunny with 3597 Samples.]{
		\label{tea}		
		\centering
		\includegraphics[height=1.2in]{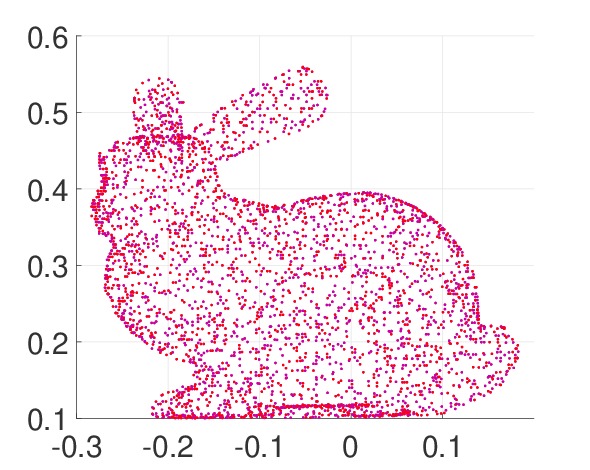}
	}	
	\subfigure[Bunny with 397 Samples.]{
		\label{bun}
		\centering
		\includegraphics[height=1.2in]{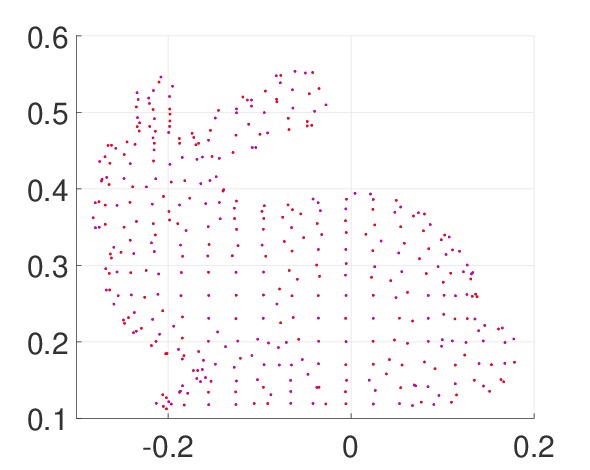}
	}%
	\centering
	\caption{Original Test Data.}
	\label{den}
\end{figure}

From estimated hypergraph spectral pairs from noisy point clouds,
the performance of denoising is an intuitive metric of how good
the estimates are. 
There are multiple methods developed to denoise
noisy point clouds. The authors of \cite{c8} proposed
a graph-based method to denoise based on total variation (GSP-TV). 
This method constructs a graph based on observed coordinates first
before solving the denoising optimization
\begin{equation}
	\min_{\mathbf{Y}} ||\mathbf{X}-\mathbf{Y}||_2^2+\alpha \mathbf{TV}(\mathbf{Y,W}),
\end{equation}
where $\mathbf{X}$ is the observed coordinates, and $\mathbf{W}$ is the adjacency matrix. Here, the graph total variation $\mathbf{TV}(\mathbf{Y,W})$ is applied in describing the smoothness over the graphs. In addition to total variation, Laplacian regularization (LR)
has also been used in denoising with a basic formulation
\begin{equation}
\min_{\mathbf{Y}} ||\mathbf{X}-\mathbf{Y}||_2^2+\alpha ||\mathbf{Y}^T\mathbf{LY}||^2_2,
\end{equation}
where $\mathbf{L}$ is the Laplacian matrix. Developed from traditional Laplacian regularization methods, a mesh Laplacian smooth (MLS) method is 
given in \cite{c37}.
\begin{table}[t]
	\begin{tabular}{|c|c|c|c|c|c|}
		\hline
		& HGSP           & GSP(TV) & MLS    & LR     & Noisy  \\ \hline
		Uniform$\sim$U(-0.03,0.03) & \textbf{32.60} & 45.94  & 56.63  & 48.86  & 63.84  \\ \hline
		Uniform$\sim$U(0.08,0.16)  & \textbf{98.36} & 160.18 & 205.15 & 168.17 & 220.96 \\ \hline
		Guassian$\sim$N(0,0.08)    & \textbf{41.10} & 42.36  & 49.41  & 64.00  & 76.54  \\ \hline
		Guassian$\sim$N(0.02,0.08) & \textbf{73.43} & 76.07  & 83.25  & 123.08 & 142.11 \\ \hline
		Impulse (p=0.08)            & \textbf{34.53} & 45.45  & 50.89  & 40.53  & 60.5   \\ \hline
	\end{tabular}
	\\
	\centering
	\caption{Error in Dfferent Kinds of Noise}
	\label{t1}
\end{table}

\begin{figure*}[htbp]
	\subfigure[GSP-based Sampling with 30\%, 50\%, 70\%, 90\%, 98\% Ratio of Samples for Cat Datasets.]{
		\label{cat2}
		\centering
		\includegraphics[width=7in]{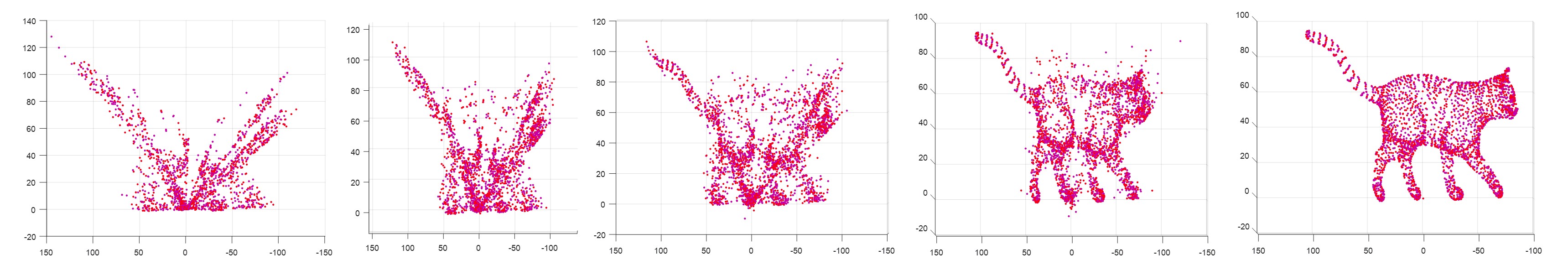}
	}%
	
	\subfigure[HGSP-based Sampling with 30\%, 50\%, 70\%, 90\%, 98\% Ratio of Samples for Cat Datasets.]{
		\label{cat3}		
		\centering
		\includegraphics[width=7in]{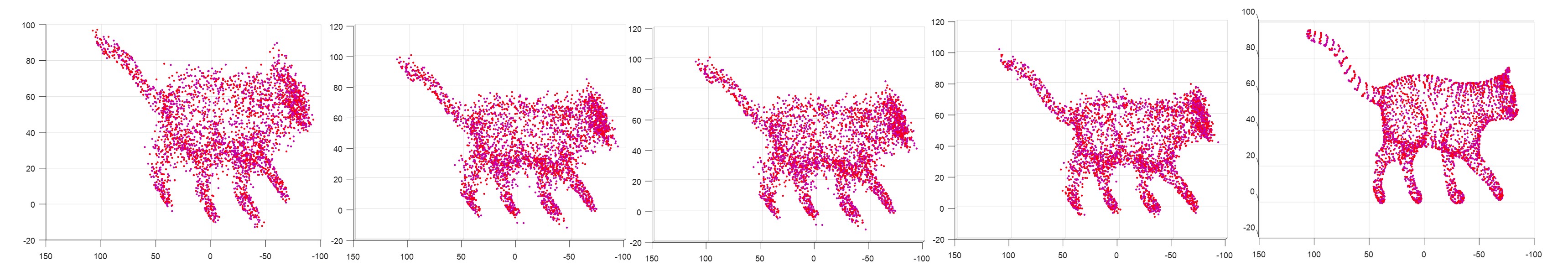}
	}	
	\subfigure[GSP-based Sampling with 30\%, 50\%, 70\%, 90\%, 98\% Ratio of Samples for Horse Datasets.]{
		\label{hor2}
		\centering
		\includegraphics[width=7in]{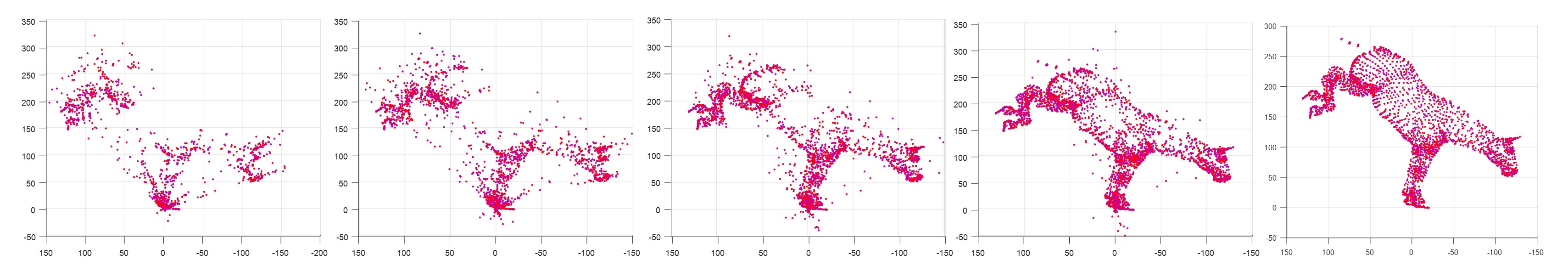}
	}%
	
	\subfigure[HGSP-based Sampling with 30\%, 50\%, 70\%, 90\%, 98\% Ratio of Samples for Horse Datasets.]{
		\label{hor3}		
		\centering
		\includegraphics[width=7in]{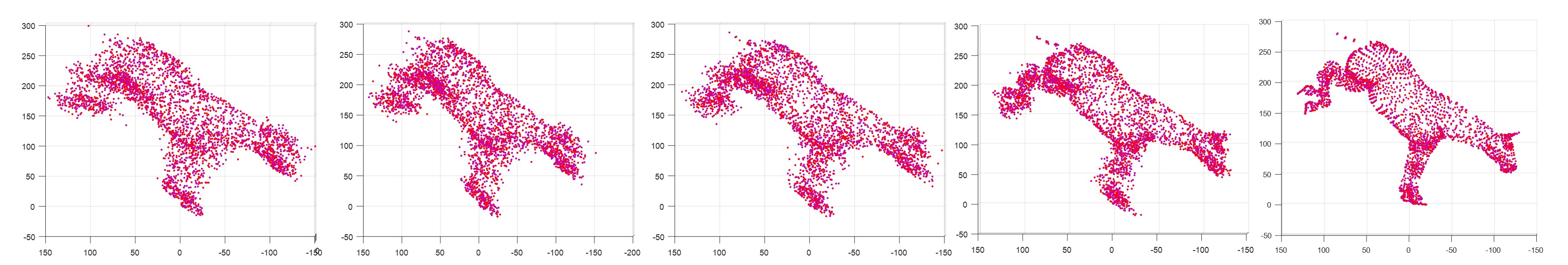}
	}	
	\subfigure[GSP-based Sampling with 30\%, 50\%, 70\%, 90\%, 98\% Ratio of Samples for Wolf Datasets.]{
		\label{wolf2}
		\centering
		\includegraphics[width=7in]{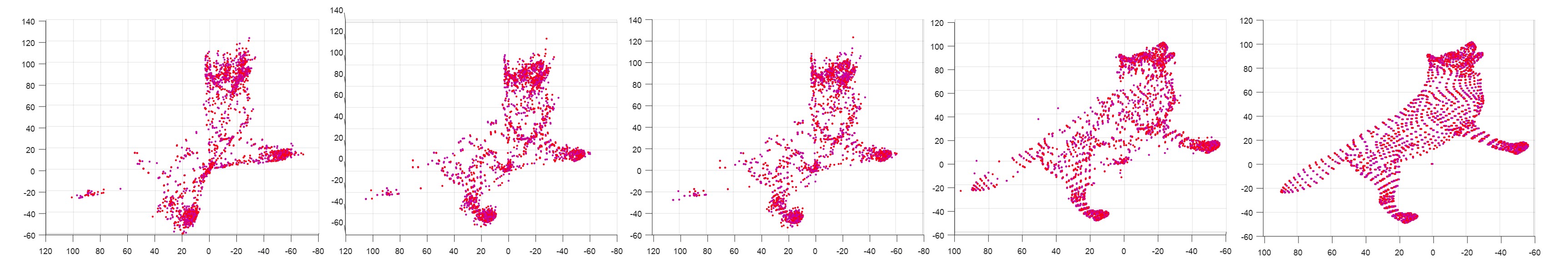}
	}%
	
	\subfigure[HGSP-based Sampling with 30\%, 50\%, 70\%, 90\%, 98\% Ratio of Samples for Wolf Datasets.]{
		\label{wolf3}		
		\centering
		\includegraphics[width=7in]{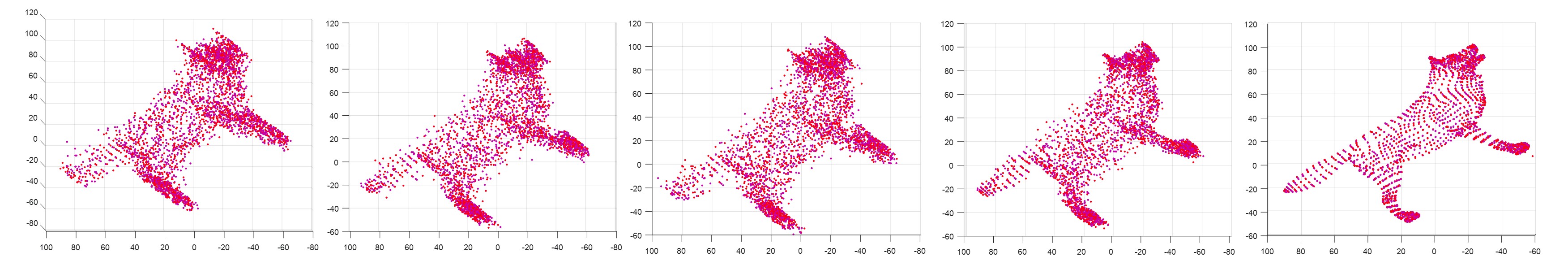}
	}
	\centering
	\caption{Recovered Point Clouds from Sampled Transformed Signals.}
	\label{sam4}
\end{figure*}

\begin{figure*}[t]
	\subfigure[Comparison in Guassian Distribution.]{
		\label{bun11}
		\centering
		\includegraphics[height=2.2in]{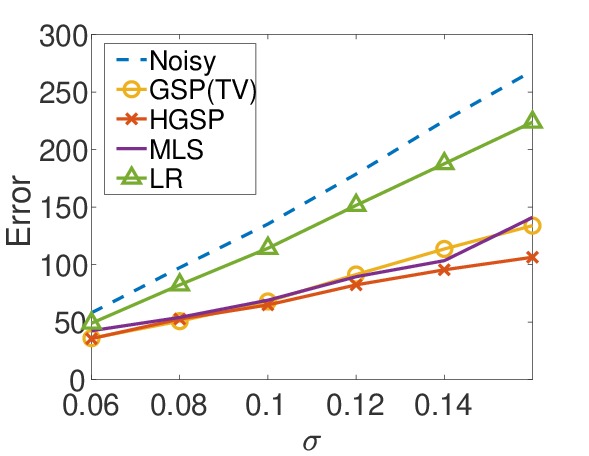}
	}%
	\hspace{0.5in}
	\subfigure[Comparison in Uniform Distribution.]{
		\label{tea11}		
		\centering
		\includegraphics[height=2.2in]{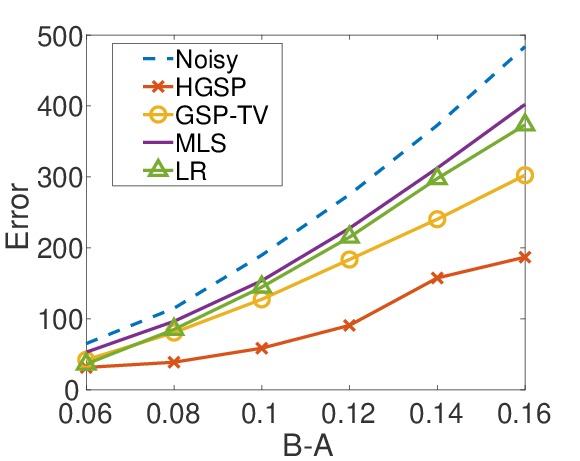}
	}	
	\centering
	\caption{Comparison between Different Methods.}
	\label{den1}
\end{figure*}

\begin{figure*}[t]
	\subfigure[Original Bunny Dataset.]{
		\label{bun2}
		\centering
		\includegraphics[height=1.4in]{original_bunny_small.jpg}
	}%
	\subfigure[Noisy Bunny Dataset.]{
		\label{nb1}		
		\centering
		\includegraphics[height=1.4in]{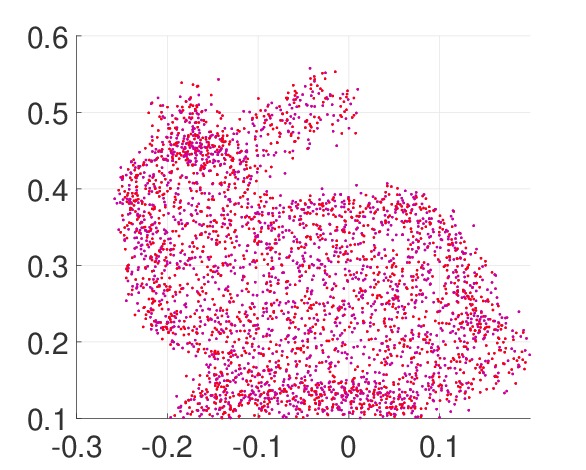}
	}	
	\subfigure[Denoised Bunny Dataset.]{
		\label{tea1}		
		\centering
		\includegraphics[height=1.4in]{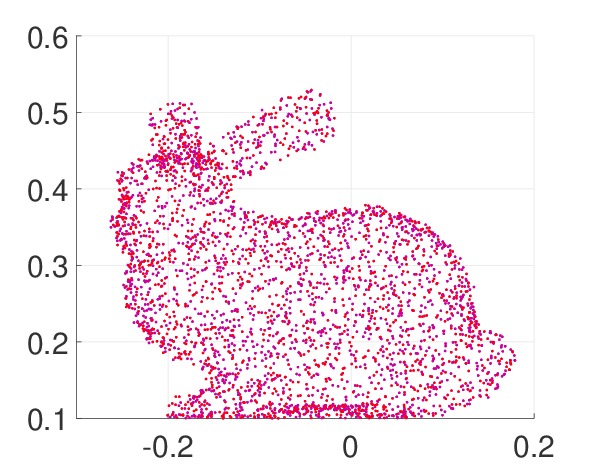}
	}	
	\centering
	\caption{Denoising of Bunny Dataset.}
	\label{den11}
\end{figure*}

To validate the performance of our denoising method, we compare 
with the aforementioned traditional methods using
the Standford bunny dataset with 3595 points and sampled bunny with 397 points shown as Fig. \ref{den}. 
We compare different methods in the sampled bunny dataset adding zero-mean Guassian noise with variance $\sigma^2$, and zero-mean Uniform noise with the interval $B-A$, respectively.
 We use the error denoted by 
\begin{equation}
	Error=\sum_{i=1}^N\sum_{j=1}^3 |X_{ji}-Y_{ji}|,
\end{equation}
where $X_{ij}$ and $Y_{ji}$ are the $j$th coordinates of observed and denoised point $i$, respectively, to measure the performance. 
We repeat the test on 1000 randomly generated noisy data.
The error
between the original dataset and the denoised dataset is shown in Fig. \ref{den1}. 
The error of the noisy point clouds before denoising is also given
as a reference in Fig. \ref{den1}. From the test results, we can see that the 
HGSP-based method can achieve the lowest error, which demonstrates the effectiveness 
of the proposed denoising methods and estimated spectral pairs. The comparison in other types of noise is shown in Table. \ref{t1}. More specifically, the methods based on total variation, i.e, HGSP and GSP-TV, have better performance than the methods based on Laplacian regularization, which indicates the total variation has a more efficient representation of the surface smoothness.
The denoised bunny with 3595 samples is shown in Fig. \ref{den11}, using our proposed 
method to denoise the noisy bunny. The successful recovery of the bunny point cloud
presents a strong evidence that our estimated spectral pairs and denoising method are
powerful tools in processing noisy datasets.

\section{Conclusions and Future Directions}\label{con}

In this work, we develop HGSP tools for effectively processing 3D point clouds.
We first introduce a novel method to estimate hypergraph spectral components 
and presented an optimization formulation to optimally select frequency coefficients to recover the optimal hypergraph structure. 
We develop a HGSP algorithm to jointly estimate hypergraph spectrum pairs
and denoise noisy point clouds. 
To test the practicality and efficacy of our proposed hypergraph tools, 
we study two point cloud  application examples.
Our results illustrate significant performance improvements for both sampling
and denoising applications. Moreover, we establish a clear connection
between hypergraph frequency components and features on point-cloud surface that
can be exploited in future studies. 

Our work establish hypergraph signal processing as an efficient tool 
in tackling high-dimensional interactions among multiple nodes. 
In addition to sampling and denoising, HGSP can 
find good applications in many other aspects of point clouds
through estimation of spectral components and frequency coefficients. 
One direction is the design of filters to analyze the spectral properties and surface features of 3D point clouds. Another interesting problem is the recovery of
point clouds from low dimensional samples. 
Beyond point clouds,  HGSP can also effectively handle 
datasets with other complex underlying structure.

\IEEEpeerreviewmaketitle

\ifCLASSOPTIONcaptionsoff
  \newpage
\fi


\end{document}